\documentclass[conference]{IEEEtran}
\IEEEoverridecommandlockouts
\usepackage{setspace}
\usepackage{cite}
\usepackage{amsmath,amssymb,amsfonts,mathtools,bm}
\usepackage{amsthm}
\usepackage{algorithm}
\usepackage{algpseudocode}
\usepackage{graphicx}
\usepackage{textcomp}
\usepackage{xcolor,float}
\usepackage{tabularx}
\usepackage{textcomp}
\usepackage{diagbox}
\usepackage{svg}
\usepackage{booktabs}
\usepackage{subfigure}
\usepackage{graphicx}
\usepackage{makecell}
\usepackage{amsthm}
\usepackage{caption}
\usepackage{sidecap}
\usepackage{microtype}
\usepackage{booktabs} 
\usepackage{multirow}
\usepackage{float}
\usepackage{adjustbox} 
\usepackage{amsmath}
\usepackage{amssymb}
\usepackage{colortbl}
\usepackage{makecell}
\usepackage{float}
\usepackage{url}
\usepackage{balance}
\usepackage{bbding}
\usepackage{graphicx}
\usepackage{sidecap}
\usepackage{microtype}
\usepackage{graphicx}
\usepackage{subfigure}
\usepackage{booktabs} 
\usepackage{multirow}
\usepackage{array}
\usepackage{subcaption}
\usepackage{graphicx}
\usepackage{makecell}
\usepackage{amsmath}
\usepackage{amssymb}
\newtheorem{proposition}{Proposition}
\newtheorem{remark}{Remark}
\usepackage{tikz}
\usepackage{graphicx}
\usepackage{caption}
\usepackage{subcaption}

\usepackage{changes} 

\usepackage{booktabs}
\usepackage{multirow}
\usepackage{adjustbox}
\usepackage{xcolor}

\definecolor{ForestGreen}{rgb}{0.13,0.55,0.13}

\definecolor{lime}{HTML}{A6CE39}
\DeclareRobustCommand{\orcidicon}{%
    \begin{tikzpicture}
    \draw[lime, fill=lime] (0,0) 
    circle [radius=0.16] 
    node[white] {{\fontfamily{qag}\selectfont \tiny ID}};    \draw[white, fill=white] (-0.0625,0.095) 
    circle [radius=0.007];    \end{tikzpicture}
    \hspace{-2mm}}
\foreach \x in {A, ..., Z}{%
    \expandafter\xdef\csname orcid\x\endcsname{\noexpand\href{https://orcid.org/\csname orcidauthor\x\endcsname}{\noexpand\orcidicon}}
    }

\allowdisplaybreaks
\renewcommand{\baselinestretch}{1}

\begin{document}

\title{Path-Level Radio Map-Aided Fast and Robust Channel Estimation for Pilot-Starved MIMO-OFDM Systems}
\author{\IEEEauthorblockN{Xiucheng Wang\IEEEauthorrefmark{1},
Nan Cheng\IEEEauthorrefmark{1},
Yiyan Zhang\IEEEauthorrefmark{2},
Ruijin Sun\IEEEauthorrefmark{1},
Haixia Peng\IEEEauthorrefmark{3}}
\IEEEauthorblockA{
\IEEEauthorrefmark{1}State Key Laboratory of ISN and School of Telecommunications Engineering, Xidian University, Xi'an, 710071, China\\
\IEEEauthorrefmark{2}School of Computer Science and Technology, Xi’an
Jiaotong University, Xi’an, China\\
\IEEEauthorrefmark{3}School of Information and Communication Engineering, Xi’an Jiaotong University, Xi’an, China\\
Email: xcwang\_1@stu.xidian.edu.cn,  dr.nan.cheng@ieee.org,  zhangyy414@stu.xjtu.edu.cn,\\ sunruijin@xidian.edu.cn, haixia.peng@xjtu.edu.cn}
}
    \maketitle

\IEEEdisplaynontitleabstractindextext

\IEEEpeerreviewmaketitle

\begin{abstract}
Accurate channel estimation in massive multiple-input multiple-output orthogonal frequency division multiplexing (MIMO-OFDM) systems is challenging when the number of pilot symbols is much smaller than the number of transmit antennas. Conventional compressed sensing methods perform a three-dimensional search over the angle-of-arrival, angle-of-departure, and delay domains, which incurs high computational cost. In this paper, we propose CHARM (channel estimation with angular-delay radio map), a framework that extracts an angular-delay power spectrum (ADPS) prior from path-level radio maps. The ADPS identifies the joint angle-of-arrival and delay support of the dominant multipath components offline, reducing the online estimation to a one-dimensional angle-of-departure search per path. A trust-region constraint is further introduced to prevent sub-grid refinement from diverging under dictionary mismatch. Simulation results show that CHARM achieves accuracy comparable to three-dimensional joint orthogonal matching pursuit (OMP) with $34.8\times$ speedup at pilot length $T \leq 4$, and that the trust-region variant degrades by only 3.7~dB under severe dictionary mismatch of 0.2~rad standard deviation, compared with 8.2~dB without the constraint.
\end{abstract}
 

\section{Introduction}
\label{sec:intro}
 
Massive multiple-input multiple-output (MIMO) combined with orthogonal frequency division multiplexing (OFDM) is a key enabler for fifth-generation (5G) and beyond wireless systems \cite{larsson2014massive, lu2014overview}. Accurate channel state information (CSI) is essential for beamforming, resource allocation, and interference management in these systems. In practice, CSI is acquired by transmitting known pilot symbols and estimating the channel from the received signals. However, conventional least-squares estimation requires the pilot length $T$ to be at least $N_t$, which becomes prohibitive in massive MIMO deployments. When the number of pilot symbols $T$ is much smaller than $N_t$, the channel estimation problem is severely underdetermined. For instance, with $N_t = 64$ transmit antennas and $T = 4$ pilot symbols, the compression ratio reaches $16\times$. This pilot-starved regime is commonly encountered in high-mobility or latency-sensitive scenarios where the coherence time is limited. Reducing the pilot overhead while maintaining estimation accuracy is therefore a fundamental challenge for massive MIMO-OFDM systems.
 
Existing channel estimation methods for the pilot-starved regime can be broadly classified into three categories. The first category relies on statistical channel priors. The linear minimum mean squared error (LMMSE) estimator with Kronecker-structured covariance matrices is a representative method in this category \cite{shen2019channel, han2020channel}. Although LMMSE is optimal under Gaussian channel assumptions, the covariance matrix must be estimated from training data and may not capture the multipath structure accurately. Moreover, its computational cost scales as $O(K N_r^2 N_t^2)$ with $K$ subcarriers and $N_r$ receive antennas, which is expensive for large antenna arrays. The second category exploits channel sparsity through compressed sensing (CS) techniques. Joint orthogonal matching pursuit (OMP) over a three-dimensional dictionary of angle of arrival (AoA), angle of departure (AoD), and delay is a widely used approach \cite{gao2015spatially, meng2011compressive}. Joint OMP does not require offline training and adapts naturally to different environments. However, its complexity scales as $O(L G_\theta G_\phi G_\tau T)$ with the number of paths $L$, dictionary sizes $G_\theta, G_\phi, G_\tau$, and pilot length $T$. The three-dimensional search becomes the computational bottleneck, especially when fine grid resolution is needed. The third category employs deep learning for channel estimation \cite{he2018deep, mo2017channel}, which can achieve strong performance but requires large labeled datasets and may generalize poorly across different environments.
 
Recent advances in ray tracing (RT) and digital twin technologies have enabled accurate site-specific channel modeling at reasonable computational cost \cite{alkhateeb2023real, zeng2024tutorial}. A ray tracing simulation produces the full set of multipath parameters, including the AoA, AoD, delay, and complex gain of each propagation path. This path-level information constitutes a radio map (RM) that captures the fine-grained multipath geometry of the environment \cite{wang2024radiodiff, 11278649}. While radio maps have been widely used for coverage prediction and localization \cite{levie2021radiounet, zeng2021toward, wang2026tutorial}, their potential for aiding channel estimation remains underexplored. Prior works that leverage ray tracing for channel estimation typically use the simulated channel directly as an initial estimate or a training label for neural networks \cite{jiang2023digital, oh2004mimo}. In contrast, our approach extracts a structured angular-delay power spectrum (ADPS) prior that enables dimensionality reduction of the search space. The key insight is that AoA and delay are receiver-side parameters independent of the transmit beamforming, whereas the AoD is coupled with the pilot matrix. The ADPS derived from a path-level radio map reveals the joint AoA-delay support of all propagation paths. If this support is known, the estimator only needs to search for the AoD of each path, reducing the search from three dimensions to one.
 
In this paper, we propose CHARM, a path-level radio map-aided channel estimation framework for pilot-starved MIMO-OFDM systems. The ADPS is first constructed offline from ray tracing data, and the dominant multipath components are extracted through peak detection with sub-grid refinement. During online estimation, the receiver projects the observations onto the extracted AoA support, compensates for the path delays, and performs a one-dimensional AoD search for each path. A trust-region constraint is further introduced to prevent the sub-grid refinement from diverging under dictionary mismatch. The main contributions of this paper are summarized as follows.
 
\begin{itemize}
    \item We propose to extract an ADPS prior from path-level radio maps, which reduces channel estimation from a three-dimensional dictionary search to a one-dimensional AoD search. The online complexity is $O(L G_\phi T + L K N_r N_t)$, where the reconstruction term dominates for small $T$, making the runtime effectively independent of the pilot length.
    \item We introduce a trust-region constraint on sub-grid refinement that confines the refined estimates within one grid bin of the original peak. We prove that this constraint bounds the AoA estimation error regardless of the dictionary bias magnitude, reducing performance degradation from 8.2~dB to 3.7~dB at negligible computational cost.
    \item Comprehensive simulations are conducted under varying pilot lengths, SNR levels, and dictionary mismatch conditions. Results demonstrate that CHARM achieves accuracy comparable to joint OMP with $34.8\times$ speedup at $T \leq 4$, and that the trust-region variant degrades by only 3.7~dB under dictionary mismatch of 0.2~rad standard deviation, compared with 8.2~dB without the constraint.
\end{itemize}
 
The remainder of this paper is organized as follows. Section~\ref{sec:system} presents the system model. Section~\ref{sec:method} describes the proposed method. Section~\ref{sec:results} provides simulation results. Section~\ref{sec:conclusion} concludes the paper.
 
\section{System Model and Problem Formulation}
\label{sec:system}
 
\subsection{MIMO-OFDM Signal Model}
 
We consider a downlink MIMO-OFDM system with $N_t$ transmit antennas and $N_r$ receive antennas, both configured as uniform linear arrays (ULAs). The system operates over $K$ subcarriers with subcarrier spacing $\Delta f$. The frequency-domain channel matrix at the $k$-th subcarrier is denoted by $\mathbf{H}[k] \in \mathbb{C}^{N_r \times N_t}$, where $k = 0, 1, \ldots, K-1$. During the channel estimation phase, the transmitter sends $T$ pilot symbols over $T$ time slots. In the $t$-th time slot, the transmit beamforming vector is $\mathbf{x}_t \in \mathbb{C}^{N_t \times 1}$ with $\|\mathbf{x}_t\|^2 = 1$. The received signal at the $k$-th subcarrier in the $t$-th time slot is given by
\begin{equation}
    \mathbf{y}[t, k] = \mathbf{H}[k] \mathbf{x}_t + \mathbf{n}[t, k],
    \label{eq:received_signal}
\end{equation}
where $\mathbf{n}[t, k] \sim \mathcal{CN}(\mathbf{0}, \sigma^2 \mathbf{I}_{N_r})$ is additive white Gaussian noise. The pilot matrix is defined as $\mathbf{X} = [\mathbf{x}_1, \mathbf{x}_2, \ldots, \mathbf{x}_T] \in \mathbb{C}^{N_t \times T}$. In the pilot-starved regime where $T \ll N_t$, the direct least-squares solution $\hat{\mathbf{H}}[k] = \mathbf{Y}[k] \mathbf{X}^H (\mathbf{X} \mathbf{X}^H)^{-1}$ is infeasible due to the rank deficiency of $\mathbf{X} \mathbf{X}^H$.
 
\subsection{Geometric Channel Model}
 
We adopt a geometric channel model with $L$ propagation paths. Each path $l$ is characterized by its complex gain $\alpha_l$, AoA $\theta_l$, AoD $\phi_l$, and propagation delay $\tau_l$. The frequency-domain channel matrix at the $k$-th subcarrier is expressed as
\begin{equation}
    \mathbf{H}[k] = \sum_{l=1}^{L} \alpha_l \, e^{-j 2\pi k \Delta f \tau_l} \, \mathbf{a}_r(\theta_l) \, \mathbf{a}_t^H(\phi_l),
    \label{eq:channel_model}
\end{equation}
where $\mathbf{a}_r(\theta_l) \in \mathbb{C}^{N_r \times 1}$ and $\mathbf{a}_t(\phi_l) \in \mathbb{C}^{N_t \times 1}$ are the receive and transmit array steering vectors, respectively. For a ULA with half-wavelength spacing, the steering vector is given by
\begin{equation}
    \mathbf{a}(\theta) = \frac{1}{\sqrt{N}} \left[1, \, e^{j\pi \sin\theta}, \, \ldots, \, e^{j\pi (N-1) \sin\theta} \right]^T,
    \label{eq:steering_vector}
\end{equation}
where $N$ denotes the number of antennas. The channel in \eqref{eq:channel_model} is parameterized by $4L$ unknowns. Since the number of resolvable paths $L$ is typically small compared to the array dimension, the channel exhibits sparsity in the angle-delay domain. This sparsity has been widely exploited by CS methods to enable channel estimation with $T \ll N_t$.
 
\subsection{Problem Formulation}
 
Substituting \eqref{eq:channel_model} into \eqref{eq:received_signal}, the received signal can be rewritten as
\begin{equation}
    \mathbf{y}[t, k] = \sum_{l=1}^{L} \alpha_l \, e^{-j 2\pi k \Delta f \tau_l} \, (\mathbf{x}_t^H \mathbf{a}_t(\phi_l)) \, \mathbf{a}_r(\theta_l) + \mathbf{n}[t, k].
    \label{eq:received_expanded}
\end{equation}
Conventional CS approaches formulate channel estimation as sparse recovery over a three-dimensional dictionary of size $G_\theta \times G_\phi \times G_\tau$. With $4\times$ oversampled dictionaries where $G_\theta = 4 N_r$ and $G_\phi = 4 N_t$, the total dictionary size reaches $G_\theta \times G_\phi \times G_\tau \approx 4.2 \times 10^6$ atoms. Joint OMP iteratively selects atoms from this dictionary, resulting in a per-iteration complexity of $O(G_\theta G_\phi G_\tau T K)$. This large search space leads to high computational cost and motivates the development of methods that can reduce the effective search dimension. We evaluate estimation accuracy using the normalized mean squared error (NMSE):
\begin{equation}
    \text{NMSE} = \frac{\sum_{k=0}^{K-1} \|\hat{\mathbf{H}}[k] - \mathbf{H}[k]\|_F^2}{\sum_{k=0}^{K-1} \|\mathbf{H}[k]\|_F^2}.
    \label{eq:nmse}
\end{equation}
 
\section{Proposed Method}
\label{sec:method}
 
The proposed CHARM framework consists of an offline phase and an online phase, as illustrated in Fig.~\ref{fig:pipeline}. In the offline phase, an ADPS is constructed from a path-level radio map, and the dominant multipath components are extracted through peak detection with sub-grid refinement. In the online phase, the extracted AoA-delay support is used as a prior to reduce the channel estimation to a one-dimensional AoD search per path.
 
\begin{figure*}[t]
    \centering
    \includegraphics[width=1\textwidth]{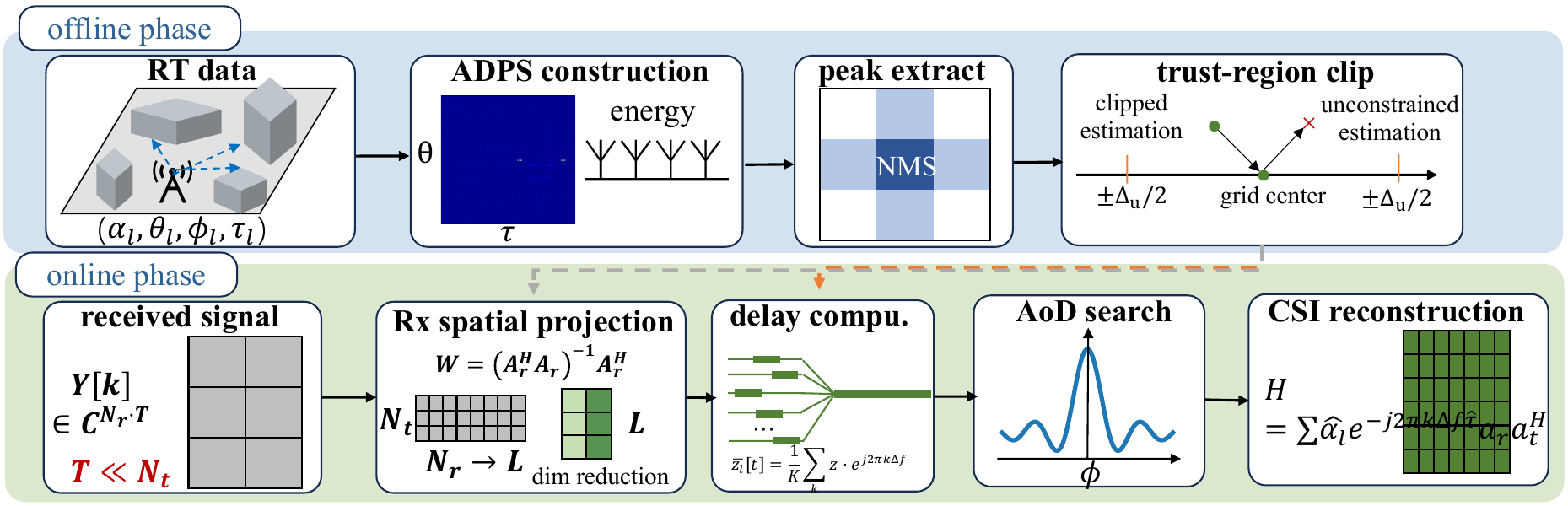}
    \caption{Overview of the proposed CHARM framework. The offline phase extracts AoA-delay support from a path-level radio map via ADPS construction, peak extraction, and trust-region constrained refinement. The online phase uses this support to perform receive-side spatial projection, delay compensation, and one-dimensional AoD search, reducing the three-dimensional dictionary search to a single dimension.}
    \label{fig:pipeline}
\end{figure*}
 
\subsection{Offline Phase: ADPS Construction and Peak Extraction}
 
\subsubsection{ADPS Construction}
 
Given a set of $L_{\text{rt}}$ ray tracing paths $\{(\alpha_l, \theta_l, \phi_l, \tau_l)\}_{l=1}^{L_{\text{rt}}}$, we construct the ADPS over a two-dimensional grid in the AoA-delay plane. The AoA axis is uniformly sampled in the $\sin\theta$ domain with an oversampling factor of 4, yielding $G_\theta = 4 N_r$ grid points. The delay axis consists of $G_\tau = K$ bins with resolution $1/(K \Delta f)$. The ADPS value at grid point $(\theta_i, \tau_j)$ is computed as
\begin{equation}
    P(\theta_i, \tau_j) = \sum_{l=1}^{L_{\text{rt}}} |\alpha_l|^2 \, \left|\mathbf{a}_r^H(\theta_i) \, \mathbf{a}_r(\theta_l)\right|^2 \, D_K^2(\tau_j - \tau_l),
    \label{eq:adps}
\end{equation}
where $D_K(\cdot)$ denotes the $K$-point Dirichlet kernel that captures the delay-domain response of the OFDM system. The ADPS encodes the AoA and delay information of all propagation paths into a two-dimensional map. Each peak in this map corresponds to a dominant multipath component, and its position indicates the AoA and delay of that component. The ADPS thus serves as a structured environmental prior that reveals the multipath support without requiring online pilot observations.
 
\subsubsection{Peak Extraction and Sub-Grid Refinement}
 
The dominant multipath components are identified by detecting local maxima on the ADPS grid. A $3 \times 3$ non-maximum suppression window is applied to locate peaks satisfying two conditions: the peak value exceeds a dynamic threshold set at 10~dB below the global maximum, and no neighboring grid point has a higher value. The number of extracted peaks is limited to $L_{\max} = \min(N_r, 16)$ to avoid overfitting. Each detected peak $l$ is recorded with its grid indices $(i_l, j_l)$, AoA $\theta_l^{\text{grid}}$, delay $\tau_l^{\text{grid}}$, and power $P_l$.
 
The grid-based peak positions are limited by the finite ADPS resolution. To achieve sub-grid accuracy, three-point parabolic interpolation is applied along both dimensions independently. For a peak at grid index $i$ with neighboring values $P_{i-1}$, $P_i$, and $P_{i+1}$, the fractional offset is
\begin{equation}
    \delta = \frac{1}{2} \cdot \frac{P_{i-1} - P_{i+1}}{P_{i-1} - 2P_i + P_{i+1}} \cdot \Delta,
    \label{eq:parabolic}
\end{equation}
where $\Delta$ is the grid spacing. The refined estimates are $\hat{\theta}_l = \theta_l^{\text{grid}} + \delta_\theta$ and $\hat{\tau}_l = \tau_l^{\text{grid}} + \delta_\tau$. This refinement recovers approximately 0.3 grid bins of sub-grid accuracy and contributes 4.6~dB of NMSE improvement, as shown in Section~\ref{sec:results}.
 
\subsection{Online Phase: Channel Estimation}
 
The online estimation uses the extracted support $\{(\hat{\theta}_l, \hat{\tau}_l)\}_{l=1}^{L}$ to reduce the problem to a one-dimensional AoD search. The key observation is that AoA and delay are receiver-side parameters that can be identified from the ADPS independently of the transmit beamforming, whereas the AoD is coupled with the pilot matrix $\mathbf{X}$ and must be estimated from the received signals.
 
\subsubsection{Receive-Side Spatial Projection}
 
The receive steering matrix is constructed as $\mathbf{A}_r = [\mathbf{a}_r(\hat{\theta}_1), \ldots, \mathbf{a}_r(\hat{\theta}_L)] \in \mathbb{C}^{N_r \times L}$. A spatial projection matrix separates the contributions from different paths. When the steering matrix is well-conditioned, the projection is $\mathbf{W} = (\mathbf{A}_r^H \mathbf{A}_r)^{-1} \mathbf{A}_r^H \in \mathbb{C}^{L \times N_r}$. When the condition number $\kappa(\mathbf{A}_r^H \mathbf{A}_r)$ exceeds 100, Tikhonov regularization is applied:
\begin{equation}
    \mathbf{W} = (\mathbf{A}_r^H \mathbf{A}_r + \lambda \, \text{diag}(1/P_l))^{-1} \mathbf{A}_r^H,
    \label{eq:projection_reg}
\end{equation}
where $P_l$ is the ADPS power of the $l$-th peak and $\lambda$ is a regularization parameter. The per-path observation at time slot $t$ and subcarrier $k$ is then $z[t, k, l] = \mathbf{w}_l^H \mathbf{y}[t, k]$, where $\mathbf{w}_l^H$ is the $l$-th row of $\mathbf{W}$.
 
\subsubsection{Delay Compensation and Subcarrier Averaging}
 
The per-path observations are coherently combined across subcarriers by compensating for the estimated delay:
\begin{equation}
    \bar{z}_l[t] = \frac{1}{K} \sum_{k=0}^{K-1} z[t, k, l] \cdot e^{+j 2\pi k \Delta f \hat{\tau}_l}.
    \label{eq:delay_comp}
\end{equation}
When the delay estimate is accurate, the signal components add coherently while the noise components add incoherently. Let $\sigma_z^2$ denote the noise variance after spatial projection. After averaging over $K$ subcarriers, the noise power reduces to $\sigma_z^2 / K$, yielding an effective SNR gain of $K = 128$. This gain is critical for the subsequent AoD search, enabling reliable correlation even with very few pilot symbols. The resulting matrix $\bar{\mathbf{Z}} \in \mathbb{C}^{L \times T}$ contains one row per path.
 
\subsubsection{One-Dimensional AoD Search and Channel Reconstruction}
 
For each path $l$, the AoD is estimated over a dictionary $\{\phi_g\}_{g=1}^{G_\phi}$ with $G_\phi = 4 N_t$. Let $\bar{\mathbf{z}}_l = [\bar{z}_l[1], \ldots, \bar{z}_l[T]]^T$ and $\mathbf{u}_g = [\mathbf{x}_1^H \mathbf{a}_t(\phi_g), \ldots, \mathbf{x}_T^H \mathbf{a}_t(\phi_g)]^T$. The AoD estimate and path gain are
\begin{equation}
    \hat{g}_l = \arg\max_{g} \frac{|\mathbf{u}_g^H \bar{\mathbf{z}}_l|^2}{\|\mathbf{u}_g\|^2}, \quad
    \hat{\alpha}_l = \frac{\mathbf{u}_{\hat{g}_l}^H \bar{\mathbf{z}}_l}{\|\mathbf{u}_{\hat{g}_l}\|^2}.
    \label{eq:aod_search}
\end{equation}
The channel is then reconstructed as
\begin{equation}
    \hat{\mathbf{H}}[k] = \sum_{l=1}^{L} \hat{\alpha}_l \, e^{-j 2\pi k \Delta f \hat{\tau}_l} \, \mathbf{a}_r(\hat{\theta}_l) \, \mathbf{a}_t^H(\hat{\phi}_l).
    \label{eq:reconstruction}
\end{equation}
 
\subsubsection{Complexity Analysis}
 
The total online complexity of CHARM is $O(L G_\phi T + L K N_r N_t)$, where the first term is the AoD search and the second is channel reconstruction. For small $T$, reconstruction dominates: with $L = 4$, $G_\phi = 256$, $T = 4$, the search cost is $4{,}096$ operations versus $1{,}048{,}576$ for reconstruction. The runtime is therefore effectively independent of $T$. In contrast, joint OMP requires $O(L_t G_\theta G_\phi G_\tau T K)$ per iteration, scaling linearly with the full three-dimensional dictionary size and the pilot length.
 
\subsection{Trust-Region Constraint}
 
In practice, the radio map may contain systematic biases due to imperfect environment modeling or outdated geometric information. Under such dictionary mismatch, the parabolic interpolation can push refined estimates far from the true values, causing up to 8.2~dB NMSE degradation. To address this, we introduce a trust-region constraint that clips the refined estimates to the vicinity of the original grid bin:
\begin{equation}
    \sin(\hat{\theta}_l) = \text{clip}\!\left(\sin(\hat{\theta}_l^{\text{ref}}), \; \sin(\theta_i^{\text{grid}}) \pm \frac{\Delta_u}{2}\right),
    \label{eq:trust_aoa}
\end{equation}
\begin{equation}
    \hat{\tau}_l = \text{clip}\!\left(\hat{\tau}_l^{\text{ref}}, \; \tau_j^{\text{grid}} \pm \frac{\Delta_\tau}{2}\right),
    \label{eq:trust_delay}
\end{equation}
where $\Delta_u$ and $\Delta_\tau$ are the grid spacings and $\text{clip}(x, a \pm b)$ constrains $x$ to $[a-b, a+b]$. We refer to this variant as CHARM~(w/~trust). The following proposition formalizes the error bound.
 
\begin{proposition}
\label{prop:trust}
Let $u_l^{\mathrm{true}} = \sin(\theta_l^{\mathrm{true}})$ and $\epsilon_{\mathrm{RT}} = |u_l^{\mathrm{grid}} - u_l^{\mathrm{true}}|$ denote the radio map bias. Under the trust-region constraint in \eqref{eq:trust_aoa}, the estimation error satisfies
\begin{equation}
    |\hat{u}_l - u_l^{\mathrm{true}}| \leq \frac{\Delta_u}{2} + \epsilon_{\mathrm{RT}}.
    \label{eq:error_bound}
\end{equation}
\end{proposition}
 
\begin{IEEEproof}
The clipping ensures $|\hat{u}_l - u_l^{\mathrm{grid}}| \leq \Delta_u / 2$. By the triangle inequality, $|\hat{u}_l - u_l^{\mathrm{true}}| \leq |\hat{u}_l - u_l^{\mathrm{grid}}| + |u_l^{\mathrm{grid}} - u_l^{\mathrm{true}}| \leq \Delta_u / 2 + \epsilon_{\mathrm{RT}}$.
\end{IEEEproof}
 
\begin{remark}
The trust-region constraint has two key properties. Under zero dictionary bias where $\epsilon_{\mathrm{RT}} = 0$, the parabolic refinement output naturally falls within the trust region, so the clipping does not activate and CHARM~(w/~trust) reduces to standard CHARM with no accuracy loss. Under large bias, the error is bounded by $\Delta_u / 2 + \epsilon_{\mathrm{RT}}$, regardless of how far the unconstrained refinement would have diverged. The same bound applies to the delay dimension with $\Delta_\tau$ replacing $\Delta_u$. The computational overhead is negligible, involving only element-wise comparisons.
\end{remark}
 
The complete procedure is summarized in Algorithm~\ref{alg:proposed}.
 
\begin{algorithm}[t]
\caption{CHARM: Proposed Channel Estimation}
\label{alg:proposed}
\begin{algorithmic}[1]
\Require Received signals $\{\mathbf{y}[t,k]\}$, pilot matrix $\mathbf{X}$, path-level radio map
\Ensure Estimated channel $\{\hat{\mathbf{H}}[k]\}_{k=0}^{K-1}$
\Statex \textbf{Offline Phase:}
\State Construct ADPS $P(\theta_i, \tau_j)$ from radio map via \eqref{eq:adps}
\State Extract peaks via $3 \times 3$ non-maximum suppression
\State Refine peaks via parabolic interpolation \eqref{eq:parabolic}
\State \textit{(Optional)} Apply trust-region via \eqref{eq:trust_aoa}--\eqref{eq:trust_delay}
\Statex \textbf{Online Phase:}
\State Construct $\mathbf{A}_r$ and compute projection $\mathbf{W}$ via \eqref{eq:projection_reg}
\For{each path $l = 1, \ldots, L$}
    \State Project: $z[t,k,l] = \mathbf{w}_l^H \mathbf{y}[t,k]$
    \State Delay compensate and average via \eqref{eq:delay_comp}
    \State 1D AoD search and gain estimation via \eqref{eq:aod_search}
\EndFor
\State Reconstruct channel via \eqref{eq:reconstruction}
\end{algorithmic}
\end{algorithm}
 
\section{Simulation Results}
\label{sec:results}
 
\subsection{Simulation Setup}
 
We consider a system with $N_t = 64$, $N_r = 32$, $K = 128$, $\Delta f = 120$~kHz, and carrier frequency 2.0~GHz, resulting in a total bandwidth of 15.36~MHz. The dictionaries have $G_\theta = 128$, $G_\phi = 256$, and $G_\tau = 128$. Channel data are generated by a ray tracing simulator that produces site-specific multipath parameters for 24 receiver locations. For each location, 8 independent trials are conducted with different noise and pilot realizations, yielding 192 paired samples per condition. Pilot vectors are drawn from the columns of a discrete Fourier transform matrix. The default operating point is $T = 4$ and $\text{SNR} = 20$~dB. CHARM is compared against Joint OMP-3D, LMMSE-Kron, Kron-OMP, and the ablation variant CHARM~(w/o~refine) that skips the parabolic interpolation step.
 
\subsection{NMSE versus Pilot Length}
 
Fig.~\ref{fig:nmse_vs_T} shows the NMSE as a function of $T$ at $\text{SNR} = 20$~dB. CHARM achieves $-12.66$~dB at $T = 4$, comparable to Joint OMP-3D at $-12.48$~dB, with only 0.18~dB difference. At $T = 2$, CHARM outperforms Joint OMP-3D by 0.9~dB ($-3.55$ vs $-2.65$~dB), demonstrating a clear advantage in the most pilot-starved condition. This advantage arises because the ADPS prior provides a strong initialization that is most beneficial when the number of pilot observations is extremely limited. The ablation variant CHARM~(w/o~refine) yields $-8.03$~dB at $T = 4$, confirming a 4.6~dB contribution from sub-grid refinement. LMMSE-Kron and Kron-OMP produce positive NMSE values across all pilot lengths, indicating that methods relying on statistical priors or per-subcarrier processing completely fail in the pilot-starved regime. When $T \geq 5$, Joint OMP-3D begins to surpass CHARM because additional pilots enable the three-dimensional search to resolve paths that the ADPS peak extraction may merge or miss. This limitation motivates the residual path recovery mechanism in our ongoing work.
 
\begin{figure}[t]
    \centering
    \includegraphics[width=0.85\linewidth]{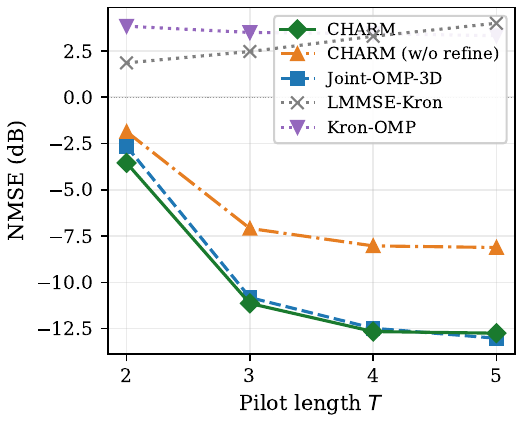}
    \caption{NMSE versus pilot length $T$ at $\text{SNR} = 20$~dB.}
    \label{fig:nmse_vs_T}
\end{figure}
 
\subsection{NMSE versus SNR}
 
Fig.~\ref{fig:nmse_vs_snr} shows the NMSE versus SNR at $T = 4$. Below $-5$~dB, all CS-based and ADPS-based methods converge because noise dominates the estimation error. Above 10~dB, CHARM saturates at $-12.65$~dB, comparable to Joint OMP-3D at $-12.52$~dB, with both methods limited by dictionary resolution and path modeling errors rather than noise. CHARM~(w/o~refine) saturates at approximately $-8.0$~dB, showing that the sub-grid refinement gain is consistent across all SNR levels. LMMSE-Kron and Kron-OMP remain at positive NMSE across the entire range, with Kron-OMP reaching $+14.4$~dB at $-15$~dB due to severe noise enhancement in the per-subcarrier processing. These results confirm that the ADPS-based framework maintains its advantage consistently without sacrificing low-SNR robustness.
 
\begin{figure}[t]
    \centering
    \includegraphics[width=0.85\linewidth]{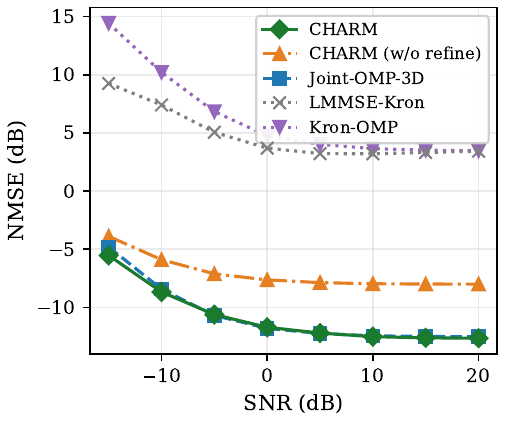}
    \caption{NMSE versus SNR at $T = 4$.}
    \label{fig:nmse_vs_snr}
\end{figure}
 
\subsection{Robustness to Dictionary Mismatch}
 
Fig.~\ref{fig:robustness} evaluates robustness under dictionary bias at $T = 4$, $\text{SNR} = 20$~dB. The bias is modeled as a random perturbation in the $\sin\theta$ domain with standard deviation from 0 to 0.2~rad. Without the trust-region constraint, CHARM degrades from $-12.46$~dB to $-4.28$~dB, a total of 8.2~dB. With the trust-region constraint, CHARM~(w/~trust) limits the degradation to 3.7~dB, maintaining $-8.80$~dB at the most severe mismatch. This $4.5$~dB improvement directly validates Proposition~\ref{prop:trust}, which guarantees that the clipping bounds the worst-case error. Joint OMP-3D maintains a constant $-12.55$~dB because it does not use the ADPS prior and is therefore immune to this type of bias. CHARM~(w/o~refine) also remains flat because it skips the parabolic interpolation entirely, though at the cost of 4.6~dB lower accuracy at zero bias. The trust-region curve saturates at bias~$\approx 0.03$~rad, after which further increases in bias do not cause additional degradation, confirming the bounded-error property.
 
\begin{figure}[t]
    \centering
    \includegraphics[width=0.85\linewidth]{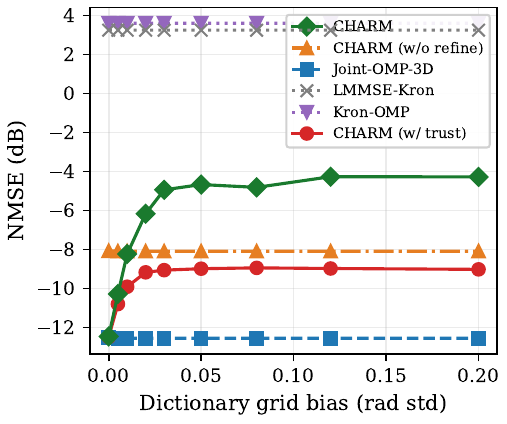}
    \caption{Robustness to dictionary mismatch at $T = 4$, $\text{SNR} = 20$~dB.}
    \label{fig:robustness}
\end{figure}
 
\subsection{Runtime and Overall Comparison}
 
Fig.~\ref{fig:runtime} shows the runtime versus $T$ on a logarithmic scale. CHARM and CHARM~(w/o~refine) remain stable at approximately 10~ms regardless of $T$, experimentally confirming that the channel reconstruction cost dominates as analyzed in Section~\ref{sec:method}. In contrast, Joint OMP-3D increases from approximately 200~ms at $T = 2$ to over 400~ms at $T = 5$, reflecting its linear dependence on $T$. Kron-OMP exceeds 1,700~ms, making it impractical for latency-sensitive applications. Table~\ref{tab:comparison} summarizes the performance at $T = 4$, $\text{SNR} = 20$~dB. CHARM achieves $-12.66$~dB in just 10.2~ms, yielding a $34.8\times$ speedup over Joint OMP-3D at comparable accuracy. CHARM~(w/o~refine) runs at similar speed but sacrifices 4.6~dB, highlighting the value of sub-grid refinement. LMMSE-Kron and Kron-OMP are not competitive in either accuracy or speed, confirming that statistical priors and per-subcarrier processing are inadequate for the pilot-starved regime.
 
\begin{figure}[t]
    \centering
    \includegraphics[width=0.85\linewidth]{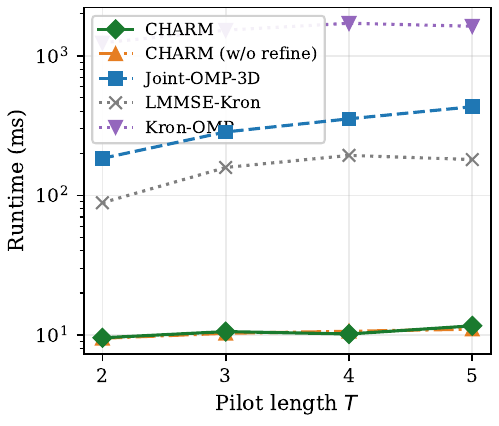}
    \caption{Runtime versus pilot length $T$.}
    \label{fig:runtime}
\end{figure}
 
\begin{table}[t]
\centering
\caption{Performance comparison at $T=4$, SNR$=20$\,dB.}
\label{tab:comparison}
\resizebox{0.98\linewidth}{!}{
\begin{tabular}{lccc}
\toprule
Method & NMSE (dB) & Runtime (ms) & Speedup \\
\midrule
CHARM & $-12.66$ & 10.2 & $34.8\times$ \\
CHARM (w/o refine) & $-8.03$ & 10.6 & $33.3\times$ \\
Joint-OMP-3D & $-12.48$ & 353.5 & $1.0\times$ \\
LMMSE-Kron & $+3.30$ & 193.3 & $1.8\times$ \\
Kron-OMP & $+3.44$ & 1711.1 & $0.2\times$ \\
\bottomrule
\end{tabular}
}
\end{table}
 
\section{Conclusion}
\label{sec:conclusion}
 
We have proposed CHARM, a path-level radio map-aided channel estimation framework that reduces the three-dimensional dictionary search to a one-dimensional AoD search by exploiting an ADPS prior. We have further introduced a trust-region constraint that bounds the sub-grid refinement error under dictionary mismatch, with a formal guarantee on the worst-case estimation error. Simulation results confirm that CHARM achieves $34.8\times$ speedup over joint OMP with comparable accuracy and that the trust-region variant limits mismatch degradation to 3.7~dB versus 8.2~dB without the constraint. The proposed framework enables fast and reliable channel acquisition in pilot-starved MIMO-OFDM systems, which is beneficial for latency-sensitive communication scenarios. Future work will extend the framework with residual path recovery to further improve accuracy when sufficient pilot symbols are available.

\bibliography{ref}
\bibliographystyle{IEEEtran}
\ifCLASSOPTIONcaptionsoff
  \newpage
\fi
\end{document}